\documentclass[10pt, letterpaper]{amsart}

\usepackage{latexsym, amsmath, amstext, amssymb, amsfonts, amscd, bm, array, multirow, amsbsy, mathrsfs}
\usepackage{amsthm}
\usepackage{t1enc}
\usepackage[mathscr]{eucal}
\usepackage{indentfirst}
\usepackage{pb-diagram}
\usepackage{graphicx}
\usepackage{fancyhdr}
\usepackage{fancybox}
\usepackage{enumerate}
\usepackage{color}
\usepackage[all]{xy}
\usepackage{hyperref}
\usepackage{tikz}
\usepackage{xparse}
\usetikzlibrary{matrix}

\usepackage{url}
\usepackage[sort&compress,comma]{natbib}
\bibpunct{[}{]}{,}{n}{}{,}
\theoremstyle{plain}
\newtheorem{thm}{Theorem}[section]

\newtheorem{prop}[thm]{Proposition}

\theoremstyle{definition}

\theoremstyle{remark}
\newtheorem{remark}[thm]{Remark}

\begin{document}

\title{A class of non-geometric M-theory compactification backgrounds} 

\author{C. S. Shahbazi}
\email{carlos.shabazi-alonso@cea.fr}
\address{Institut de Physique Th\'eorique, CEA-Saclay, France.}

\date{\today}
%\thanks{2000 Mathematics Subject Classification:   primary,  secondary   .}

\begin{abstract}
We study a particular class of supersymmetric M-theory eight-dimensional non-geometric compactification backgrounds to three-dimensional Minkowski space-time, proving that the global space of the non-geometric compactification is still a differentiable manifold, although with very different geometric and topological properties respect to the corresponding standard M-theory compactification background: it is a compact complex manifold admitting a K\"ahler covering with deck transformations acting by holomorphic homotheties with respect to the K\"ahler metric. We show that this class of non-geometric compactifications evade the Maldacena-Nu\~nez no-go theorem by means of a mechanism originally developed by Mario Garc\'ia-Fern\'andez and the author for Heterotic Supergravity, and thus do not require $l_{P}$-corrections to allow for a non-trivial warp factor or four-form flux. We obtain an explicit compactification background on a complex Hopf four-fold that solves all the equations of motion of the theory. We also show that this class of non-geometric compactifications is equipped with a holomorphic principal torus fibration over a projective K\"ahler base as well as a codimension-one foliation with nearly-parallel $G_{2}$-leaves, making thus contact with the work of M. Babalic and C. Lazaroiu on the foliation structure of the most general M-theory supersymmetric compactifications.
\end{abstract}

%\thanks{
%}

\maketitle

\setcounter{tocdepth}{1} %doesn't display subsections in TOC 
\tableofcontents

%%%%%%%%%%%%%%%%%%%%%%%%%%%%%%%%%%%%%%%%%%%%%%%%%%%%%%%%%%%%%%%%%%%%%%%%%%%%%%%%%%%%%%%%%%%%%%%%%%%%%%%%%%%%%%%%%%%%%%%%%%%%
%%%%%%%%%%%%%%%%%%%%%%%%%%%%%%%%%%%%%%%%%%%%%%%%%%%%%%%%%%%%%%%%%%%%%%%%%%%%%%%%%%%%%%%%%%%%%%%%%%%%%%%%%%%%%%%%%%%%%%%%%%%%
%%%%%%%%%%%%%%%%%%%%%%%%%%%%%%%%%%%%%%%%%%%%%%%%%%%%%%%%%%%%%%%%%%%%%%%%%%%%%%%%%%%%%%%%%%%%%%%%%%%%%%%%%%%%%%%%%%%%%%%%%%%%
%%%%%%%%%%%%%%%%%%%%%%%%%%%%%%%%%%%%%%%%%%%%%%%%%%%%%%%%%%%%%%%%%%%%%%%%%%%%%%%%%%%%%%%%%%%%%%%%%%%%%%%%%%%%%%%%%%%%%%%%%%%%

\section{Introduction}
\label{sec:introduction}

%%%%%%%%%%%%%%%%%%%%%%%%%%%%%%%%%%%%%%%%%%%%%%%%%%%%%%%%%%%%%%%%%%%%%%%%%%%%%%%%%%%%%%%%%%%%%%%%%%%%%%%%%%%%%%%%%%%%%%%%%%%%
%%%%%%%%%%%%%%%%%%%%%%%%%%%%%%%%%%%%%%%%%%%%%%%%%%%%%%%%%%%%%%%%%%%%%%%%%%%%%%%%%%%%%%%%%%%%%%%%%%%%%%%%%%%%%%%%%%%%%%%%%%%%
%%%%%%%%%%%%%%%%%%%%%%%%%%%%%%%%%%%%%%%%%%%%%%%%%%%%%%%%%%%%%%%%%%%%%%%%%%%%%%%%%%%%%%%%%%%%%%%%%%%%%%%%%%%%%%%%%%%%%%%%%%%%
%%%%%%%%%%%%%%%%%%%%%%%%%%%%%%%%%%%%%%%%%%%%%%%%%%%%%%%%%%%%%%%%%%%%%%%%%%%%%%%%%%%%%%%%%%%%%%%%%%%%%%%%%%%%%%%%%%%%%%%%%%%%

In this note we are going to consider a \emph{simple} class of non-geometric compactification backgrounds, building up on the results of reference \cite{Shahbazi:2015mva}. By non-geometric solution we mean here a global space obtained by patching up local solutions to the equations of motion using local diffeomorphisms, gauge transformations, and global symmetries of the equations of motion, namely U-dualities. The term \emph{non-geometric} may be somewhat misleading, especially for mathematicians interested in the topic, since, although there is of course no guarantee for the global space of a non-geometric solution to be a smooth differentiable manifold, it will be for sure a well-defined mathematical object, with well-defined topological an geometric properties. We will anyway use it as it has become standard in the the literature. 

Non-geometric compactification backgrounds have been intensively studied in the literature from different points of view, see for example \cite{Hellerman:2002ax,Hull:2004in,Shelton:2005cf,Wecht:2007wu,Grana:2008yw,Andriot:2013txa} for more details and further references. In particular, the approach to non-geometry through Generalized Complex Geometry \cite{2002math......9099H,2004math......1221G} and Double Field Theory \cite{Siegel:1993th,Hull:2009mi} has proved to be very efficient \cite{Grana:2008yw,Andriot:2011uh,Andriot:2012an,Blumenhagen:2012nk,Blumenhagen:2012nt,Andriot:2012wx,Blumenhagen:2013aia,Andriot:2014qla,Blumenhagen:2014gva,Hassler:2014sba,Lee:2015xga}, since these frameworks naturally include at least part of the dualities of String Theory as their automorphism group. To mention other approaches to non-geometry relevant for this letter, references \cite{McOrist:2010jw,Malmendier:2014uka,Gu:2014ova,2015arXiv150801193L} consider compactifications that are non-geometric from the Heterotic point of view and that become geometric compactifications via duality with F-theory. References \cite{Martucci:2012jk,Braun:2013yla,Candelas:2014jma,Candelas:2014kma} contain a very interesting approach, named there \emph{G-theory}, along the main idea of this work: among other things, they provide a very detailed construction of non-perturbative Type-II vacua by gluing local solutions to the equations of motion using different types of U-dualities. When performing such a non-trivial global patching, it is usually very difficult to obtain precise results about the topological and geometric properties of the global non-geometric space of the compactification. This is partly due to the fact that the symmetries of the local equations of motion involved in the global patching can be relatively involved. That is why here we will consider the arguably simplest non-geometric global patching of local solutions to the equations of motion of eleven-dimensional Supergravity on a warped compactification background to three-dimensional Minkowski space-time. In exchange for the \emph{simplicity} of the U-dualities involved in the global patching, we will be able to fully characterize the topology and the geometry of the global space.

More precisely, we will consider local solutions, more concretely local warped compactification backgrounds to three-dimensional Minkowski space, to the eleven-dimensional Supergravity equations of motion and we will globally patch them using local diffeomorphisms, gauge transformations and the \emph{trombone symmetry} of the warp factor, which simply consists on rescalings of the warp factor by a real constant. Warped compactifications of M-theory preserving different amounts of supersymmetry have been study in the literature, see for example references \cite{Becker:1996gj,Becker:2000jc,Tsimpis:2005kj,Condeescu:2013hya,Prins:2013wza,Prins:2014ona,Gutowski:2014ova}. Recently, the most general scenario was studied in references \cite{Babalic:2014fua,Babalic:2014dea,Babalic:2015xia,Babalic:2015kka}, where it was shown that a very efficient mathematical tool to study such manifolds is the theory of foliations with a particular $G$-structure on the leaves.

The global symmetry of the equations of motion that we will use to patch the local solutions is arguably the simplest one available. The idea is to consider the simplest non-geometric scenario in order to be able to fully characterize topologically as well as geometrically the global space of the compactification, something of utmost importance in order to properly understand the moduli space of a non-geometric compactification space, as well as its global properties. Hence, we hope that this compactification background will hep to understand the nature of non-geometric compactification spaces, starting from the simplest case. In fact, we will be able to show that the global space of the compactification is a differentiable manifold, but with topological and geometric properties drastically different from the corresponding standard geometric compactification background.  

Let us be more precise. In this letter we will prove, among other things, that:

\begin{itemize}

\item The non-geometric compactification space $M$ is a particular class of compact complex manifolds admitting a K\"ahler covering with deck transformations acting by holomorphic homotheties with respect to the K\"ahler metric. In other words, $M$ is a particular type of locally conformally K\"ahler manifold. Therefore, $M$ admits a K\"aler covering $\tilde{M}$ with K\"ahler form $\tilde{\omega}$, fitting in the following short sequence:

\begin{equation}
\Gamma\to \tilde{M}\to M\, .
\end{equation}

\noindent
The non-geometric warp factor is encoded in the geometry of $M$ in an elegant way. Given a $2d$-dimensional locally conformally K\"ahler manifold $(M,\omega,\theta)$ with K\"ahler form $\omega$ and closed Lee-form $\theta$, let $L$ the trivializable flat line bundle associated to the representation $A\to |\det\, A|^{\frac{1}{d}}, A\in Gl(2d, \mathbb{R}) $, with a flat connection $\nabla_{\theta} \equiv d + \theta$. The line bundle $L$ is usually called the weight bundle of $M$ and its holonomy coincides with the character $\chi\colon \pi_{1}(M)\to \mathbb{R}^{+}$. The image of $\chi$ is called the monodromy group of $M$. The warp factor is given by a flat connection of $L$ which, after choosing a trivialization, is given by a closed one-form on $M$. If $M$ is simply-connected its holonomy is trivial and then $M$ becomes a K\"ahler manifold and the compactification becomes geometric. 

\item The non-geometry of the solution is associated to the global space being non-simply-connected. If we take $M$ to be simply connected, then $M$ becomes a K\"ahler manifold and we obtain a standard geometric solution.

\item We obtain an explicit solution, preserving locally $\mathcal{N}=2$ supersymmetry, on a complex Hopf four-fold that solves all the local equations of motion of the theory, including the equation of motion for the wrap factor. We explicitly write the local metric, flux and warp factor. 

\item The previous solution evades the Maldacena-Nu\~nez theorem, by a mechanism originally developed by Mario Garc\'ia-Fern\'andez and the author for Heterotic Supergravity in reference \cite{twistHeterotic}, and thus there are non-geometric solutions with non-zero warp factor and flux without the need of higher-derivative corrections. 

\item The explicit solution on the complex Hopf four-fold is equipped with a holomorphic elliptic fibration over a K\"ahler base. This points out to a possible application of this backgrounds to F-theory compactifications.

\item The explicit solution on the complex Hopf four-fold admits a codimension-one foliation equipped with a nearly-parallel $G_{2}$-structure on the leaves. Therefore the solution, even being non-geometric, preserves the structure of the most general geometric compactification background of eleven-dimensional Supergravity on an eight-manifold, thoroughly studied in references \cite{Babalic:2014fua,Babalic:2014dea,Babalic:2015xia,Babalic:2015kka}.

\end{itemize}

\noindent
In addition, the moduli space of locally conformally K\"ahler manifolds is usually very restricted, so compactification on this backgrounds may partially evade the moduli-stabilization problem which is present in many String Theory compactifications. 

To summarize, we think that the type of non-geometric background studied in this letter is simple enough to be manageable, in particular it is possible to study its global topological and geometric properties, yet it is an honest non-trivial non-geometric compactification background. Therefore, it might be a good starting point to a systematic rigorous global study of non-geometric String Theory backgrounds. We hope this letter is a first small step in that direction. 

The outline of the note goes as follows. In section \ref{sec:Mtheory3d} we will outline the structure of supersymmetric eleven-dimensional Supergravity solutions, pointing out in a precise way the well-known \emph{issue} of imposing at the same time the classical Killing spinor equations and the $l_{P}$-corrected equations of motion, an issue that is not present in the non-geometric setting since the Maldacena-Nu\~nez no-go theorem does no hold and thus there is no need of considering $l_{P}$-corrections in order to have non-trivial solutions. In section \ref{sec:Nogeometric} we construct the non-geometric compactification and we obtain an explicit solution to all the equations of motion, studying some of its properties.\newline

%%%%%%%%%%%%%%%%%%%%%%%%%%%%%%%%%%%%%%%%%%%%%%%%%%%%%%%%%%%%%%%%%%%%%%%%%%%%%%%%%%%%%%%%%%%%%%%%%%%%%%%%%%%%%%%%%%%%%%%%%%%%
%%%%%%%%%%%%%%%%%%%%%%%%%%%%%%%%%%%%%%%%%%%%%%%%%%%%%%%%%%%%%%%%%%%%%%%%%%%%%%%%%%%%%%%%%%%%%%%%%%%%%%%%%%%%%%%%%%%%%%%%%%%%
%%%%%%%%%%%%%%%%%%%%%%%%%%%%%%%%%%%%%%%%%%%%%%%%%%%%%%%%%%%%%%%%%%%%%%%%%%%%%%%%%%%%%%%%%%%%%%%%%%%%%%%%%%%%%%%%%%%%%%%%%%%%
%%%%%%%%%%%%%%%%%%%%%%%%%%%%%%%%%%%%%%%%%%%%%%%%%%%%%%%%%%%%%%%%%%%%%%%%%%%%%%%%%%%%%%%%%%%%%%%%%%%%%%%%%%%%%%%%%%%%%%%%%%%%

\noindent{\bf Acknowledgements:} I would like to thank Mario Garc\'ia-Fern\'andez, Mariana Gra\~na, Liviu Ornea, Raffaele Savelli, Alessandro Tomasiello and Victor Vuletescu for very useful discussions and comments. I would like to thank the ICMAT, Luis \'Alvarez-Consul and Mario Garc\'ia-Fern\'andez, for the hospitality during the later stages of this project. This work was supported in part by the ERC Starting Grant 259133 -- ObservableString.

%%%%%%%%%%%%%%%%%%%%%%%%%%%%%%%%%%%%%%%%%%%%%%%%%%%%%%%%%%%%%%%%%%%%%%%%%%%%%%%%%%%%%%%%%%%%%%%%%%%%%%%%%%%%%%%%%%%%%%%%%%%%
%%%%%%%%%%%%%%%%%%%%%%%%%%%%%%%%%%%%%%%%%%%%%%%%%%%%%%%%%%%%%%%%%%%%%%%%%%%%%%%%%%%%%%%%%%%%%%%%%%%%%%%%%%%%%%%%%%%%%%%%%%%%
%%%%%%%%%%%%%%%%%%%%%%%%%%%%%%%%%%%%%%%%%%%%%%%%%%%%%%%%%%%%%%%%%%%%%%%%%%%%%%%%%%%%%%%%%%%%%%%%%%%%%%%%%%%%%%%%%%%%%%%%%%%%
%%%%%%%%%%%%%%%%%%%%%%%%%%%%%%%%%%%%%%%%%%%%%%%%%%%%%%%%%%%%%%%%%%%%%%%%%%%%%%%%%%%%%%%%%%%%%%%%%%%%%%%%%%%%%%%%%%%%%%%%%%%%

\section{About the consistency of M-theory on eight-manifolds}
\label{sec:Mtheory3d}

%%%%%%%%%%%%%%%%%%%%%%%%%%%%%%%%%%%%%%%%%%%%%%%%%%%%%%%%%%%%%%%%%%%%%%%%%%%%%%%%%%%%%%%%%%%%%%%%%%%%%%%%%%%%%%%%%%%%%%%%%%%%
%%%%%%%%%%%%%%%%%%%%%%%%%%%%%%%%%%%%%%%%%%%%%%%%%%%%%%%%%%%%%%%%%%%%%%%%%%%%%%%%%%%%%%%%%%%%%%%%%%%%%%%%%%%%%%%%%%%%%%%%%%%%
%%%%%%%%%%%%%%%%%%%%%%%%%%%%%%%%%%%%%%%%%%%%%%%%%%%%%%%%%%%%%%%%%%%%%%%%%%%%%%%%%%%%%%%%%%%%%%%%%%%%%%%%%%%%%%%%%%%%%%%%%%%%
%%%%%%%%%%%%%%%%%%%%%%%%%%%%%%%%%%%%%%%%%%%%%%%%%%%%%%%%%%%%%%%%%%%%%%%%%%%%%%%%%%%%%%%%%%%%%%%%%%%%%%%%%%%%%%%%%%%%%%%%%%%%

In this note we are interested in a particular class of non-geometric M-theory compactification backgrounds to $\mathcal{N}=2$ three-dimensional Minkowski space-time. These type of non-geometric compactifications will be introduced in section \ref{sec:Nogeometric}. In this section we will consider instead standard  M-theory supersymmetric solutions, in order to motivate how the non-geometric version of these solutions may be useful in evading some of the \emph{issues} present in the standard M-theory supersymmetric compactification scenario, such as the Maldacena-Nu\~nez no-go theorem \cite{Maldacena:2000mw}. The effective, low-energy, description of M-theory \cite{Witten:1995ex} is believed to be given by eleven-dimensional $\mathcal{N}=1$ Supergravity \cite{Cremmer:1978km}, which we will formulate on an eleven-dimensional, oriented, spinnable, differentiable manifold\footnote{By differentiable manifold we mean a Hausdorff, second-countable, topological space equipped with a \emph{differentiable structure}.} $M$. We will denote by $S\to M$ the corresponding spinor bundle, which is a bundle of $Cl(1,10)$ Clifford modules. At each point $p\in M$ we thus have that $S_{p}$ is a thirty two real, symplectic, $Cl(1,10)$ Clifford module\footnote{There are two $Cl(1,10)$ Clifford modules, which can be distinguished by the action of the volume form of $Cl(1,10)$.}, with symplectic form that we will denote by $\omega$.

The field content of eleven-dimensional Supergravity is given by a Lorentzian metric $\mathsf{g}$, a closed four-form $\mathsf{G}\in\Omega^{4}_{cl}\left(\mathcal{M}\right)$ and a Majorana gravitino $\Psi\in\Gamma\left(S\otimes \Lambda^{1}\left(\mathcal{M}\right)\right)$. We will focus only on bosonic solutions $(M,\mathsf{g},\mathsf{G})$ of the theory, so we will truncate the gravitino. The classical bosonic equations of motion are given by:

\begin{eqnarray}
\label{eq:eqsmotion}
E_{0} = \mathrm{Ric} -\frac{1}{2}\mathsf{G}\circ\mathsf{G} + \frac{1}{6}\mathsf{g}\left|\mathsf{G}\right|^2 =  0\, ,\qquad F_{0} = d\ast\mathsf{G} + \mathsf{G}\wedge\mathsf{G}  =  0\, ,
\end{eqnarray}

\noindent
where

\begin{equation}
(\mathsf{G}\circ\mathsf{G})(v,v) = |\iota_{v}\mathsf{G}|^2\, , \qquad v \in \mathfrak{X}(M)\, ,
\end{equation}

\noindent
is a symmetric $(2,0)$-tensor. Eleven-dimensional Supergravity supersymmetric bosonic solutions, and in particular supersymmetric compactification backgrounds, are defined as being solutions of eleven-dimensional Supergravity admitting at least one real spinor $\epsilon \in \Gamma(S)$ such that:

\begin{equation}
\label{eq:KSE}
D\epsilon = 0\, ,
\end{equation}

\noindent
where $D$ is the Supergravity connection acting on the bundle of Clifford-modules $S$. It is given by:

\begin{equation}
\label{eq:N1susy11}
D_{v}\epsilon\equiv \nabla_{v}\epsilon + \frac{1}{6} \iota_{v} \mathsf{G}\cdot \epsilon +\frac{1}{12} v^{\flat}\wedge \mathsf{G}\cdot\epsilon\, .
\end{equation} 

\noindent
Here $\nabla$ is the spin connection induced from the Levi-Civita connection on the tangent bundle and $\cdot$ denotes the Clifford action of forms on sections of $S$.

A supersymmetric configuration $(M,\mathsf{g},\mathsf{G})$, namely a manifold admitting a $D$-constant spinor, does not necessarily solves the eleven-dimensional Supergravity equations of motion, but it is in some sense not far from being a solution, since the integrability condition of \eqref{eq:KSE} can be written in terms of the equations of motion of the theory. The integrability condition of \eqref{eq:KSE} can be found to be:

\begin{equation}
\qquad \iota_{v}\, E_{0}\cdot \epsilon - \frac{1}{6\cdot 3!} v^{\flat}\wedge (\ast F_{0})\cdot\epsilon + \frac{1}{3!} \iota_{v}(\ast F_{0})\cdot\epsilon  = 0\, ,
\end{equation}

\noindent
where $E_{0}$ denotes the Einstein equation and $F_{0}$ denotes the Maxwell equation of eleven-dimensional Supergravity, see equations \eqref{eq:eqsmotion}. Supersymmetric solutions of eleven-dimensional Supergravity can be divided in two classes, the time-like class and the null class, see references \cite{Gauntlett:2002fz,Gauntlett:2003wb}, where the full classification of supersymmetric solutions of eleven-dimensional Supergravity was first obtained. The time-like class is given by the supersymmetric solutions that satisfy:

\begin{equation}
g(\xi^{\flat},\xi^{\flat}) > 0\, , \qquad \xi(v) = \omega(\epsilon, v\cdot\epsilon)\, ,
\end{equation}

\noindent
where $\xi$ is the one-form associated to $\epsilon$. Null supersymmetric solutions on the other hand, are those that satisfy $g(v,v) = 0$. For time-like configurations, it can be shown that if the Maxwell equation is satisfied, then the Einstein equations follow from the integrability condition of the Killing spinor equation \eqref{eq:KSE}, see reference \cite{Gauntlett:2002fz}. In other words, the Einstein equations follow from time-like supersymmetry together with the Maxwell equations of motion. Hence, as it is well known in the literature, supersymmetry is closely related to the equations of motion but it does no \emph{always} imply them. Notice that supersymmetric compactification backgrounds are indeed time-like supersymmetric solutions of eleven-dimensional Supergravity.

Compactification backgrounds of eleven-dimensional Supergravity are subject to the Maldacena-Nu\~nez no-go theorem \cite{Maldacena:2000mw}, which we state here for completeness, applied to eleven-dimensional Supergravity.

\begin{thm}
Every warped compactification of eleven-dimensional Supergravity on a closed manifoldn has constant warp factor and zero four-form flux $G$.
\end{thm}

\noindent
Therefore it would seem that if we want to define F-theory compactifications through eleven-dimensional Supergravity compactifications on an eight-dimensional manifold we will end-up having only the \emph{trivial} flux-less solution. The standard way to evade the Maldacena-Nu\~nez theorem is to include in the theory higher-derivative corrections and/or negative-tension objects. Since it is not clear whether negative-tension objects exist in M-theory, the strategy of reference \cite{Becker:1996gj} was to include the particular higher-derivative correction to eleven-dimensional Supergravity which was known at the time and which gives a negative contribution to the energy-momentum tensor of the theory. This correction was computed for the first time in the one obtained in reference \cite{Duff:1995wd}. By means of M/F-Theory duality, higher-derivative corrections to M-theory and negative-tension objects in String Theory are dual manifestations of the same phenomena \cite{Giddings:2001yu} \footnote{We thank JHEP's referee for explaining this point.}. The only dimension-full parameter in eleven-dimensional Supergravity is the Planck length $l_{P}$ and the higher-derivative corrections of M-theory arise in an expansion in powers of this constant over the relevant length-scale of the the problem under consideration. For example, the higher-derivative term considered in \cite{Becker:1996gj} is a $l^{6}_{P}$-correction. For simplicity from now on we will refer to the higher-derivative corrections of M-theory as \emph{$l_{P}$-corrections}. 

The correction to the Killing spinor equation \eqref{eq:KSE} corresponding to the correction considered in \cite{Becker:1996gj} is not known, so the analysis performed in \cite{Becker:1996gj} uses the classical Killing spinor equations and at the same time imposes $l_{P}$-corrected equations of motion. This immediately runs into a possible inconsistency, since classical supersymmetry is consistent with the classical equations of motion through the integrability condition of the Killing spinor equation, so imposing $l_{P}$-corrected equations of motion on a classical supersymmetric configuration leads to extra constraints that make the problem over determined. The possible inconsistency can be computed explicitly. Let $E$ and $F$ denote the $l_{P}$-corrected Einstein and Maxwell equations of motion. They can be written as:

\begin{equation}
E = E_{0} + E_{1}\, , \qquad F = F_{0} + F_{1}\, ,
\end{equation}

\noindent
where $E_{1}$ and $F_{1}$ denote the corresponding corrections to the classical equations of motion $E_{0}$ and $F_{0}$ and include the appropriate $l_{P}$ factors. Now, in order to study the consistency of imposing the $l_{P}$-corrected equations of motion $E$ and $F$ as well as classical supersymmetry, we only have to assume that we indeed have a solution of $l_{P}$-corrected equations of motion and compute what is the extra-constraint that appears when imposing the integrability condition of the classical Killing spinor equation. The result is, for every $v\in\mathfrak{X}(M)$, given by:

\begin{equation}
\label{eq:extraconstraint}
\iota_{v}\, E_{1}\cdot \epsilon - \frac{1}{6\cdot 3!} v^{\flat}\wedge (\ast F_{1})\cdot\epsilon + \frac{1}{3!} \iota_{v}(\ast F_{1})\cdot\epsilon  = 0\, .
\end{equation} 

\noindent 
Therefore, if we want a solution of the classical Killing spinor equation to be a solution of the $l_{P}$-corrected equations of motion, the constraint \eqref{eq:extraconstraint} must be necessarily satisfied.

The outcome of the analysis of reference \cite{Becker:1996gj} is that classical supersymmetry imposes the manifold to be a Calabi-Yau four-fold, although the physical metric does not correspond to the Ricci-flat Calabi-Yau metric. Strictly speaking then, if we want to have a solution of the $l_{P}$-corrected equations of motion, not every such Calabi-Yau is an admissible compactification background: only those satisfying equation \eqref{eq:extraconstraint}, if any, should be considered as honest solutions of the equations of motion. Let us be more explicit for the case considered in reference \cite{Becker:1996gj}. In reference \cite{Becker:1996gj} the equations of motion of classical eleven-dimensional Supergravity were modified by the only known $l_{P}$-correction at the time, obtained in reference \cite{Duff:1995wd}, and which only affects the equation of motion for $\mathsf{G}$. Hence, $E_{1} = 0$ and $F_{1}$ is given by:

\begin{equation}
\label{eq:X8}
F_{1} = \beta X_{8} = \beta\left(p^{2}_{1} - p_{2}\right)\, ,
\end{equation}

\noindent
where $p_{1}$ and $p_{2}$ are respectively the first and second Pontryagin classes of $M$, and $\beta$ is an appropriate constant. Plugging equation \eqref{eq:X8} into equation \eqref{eq:extraconstraint} we obtain the explicit constraint that the Calabi-Yau four-folds coming out of the supersymmetry analysis of \cite{Becker:1996gj} have to satisfy in order to be an honest solution of the corrected equations of motion:

\begin{equation}
\label{eq:extraconstraintX8}
\left(- v^{\flat}\wedge (\ast X_{8}) + 6\,\iota_{v}(\ast X_{8})\right)\cdot\epsilon = 0 \, .
\end{equation}

\noindent
Hence, and again strictly speaking, equation \eqref{eq:extraconstraintX8} constrains the class of admissible F-theory compactification manifolds. Admissible in the sense of honestly solving the equations of motion of $l_{P}$-corrected eleven-dimensional Supergravity and at the same time satisfying the classical Killing spinor equation of eleven-dimensional Supergravity. Of course, this \emph{problem} is well-known to experts on the field, but unfortunately, as long as the eleven-dimensional Supergravity $l_{P}$-corrected Killing spinor equation is not known, it seems not possible to solve it in a completely rigorous way. Important steps in this direction have been made in references \cite{Grimm:2014xva,Grimm:2014efa,Grimm:2015mua}, where a thoroughly and consistent analysis of M-theory compactifications in the presence of $l_{P}$-corrections has been made, and even an educated guess for the $l_{P}$-corrected Killing spinor equation has been proposed. Remarkably enough, the integrability condition of the $l_{P}$-corrected proposal for the Killing spinor equation is compatible with the $l_{P}$-corrected equations of motion, which definitely suggests that if the educated guess is not already the correct $l_{P}$-corrected Killing spinor equation, it cannot be far from being it. One of the main conclusions of \cite{Grimm:2014efa} is that even when one consistently takes into account $l_{P}$-corrections, the internal manifold of the compactification is still topologically a Calabi-Yau four-fold. This strongly suggests that the conclusion of reference \cite{Becker:1996gj} is solid after properly taking into account $l_{P}$-corrections. 

A possible, temporary, solution to the problem of imposing classical supersymmetry and $l_{P}$-corrected equations of motion, would be to consider only the elliptically fibered Calabi-Yau four-folds, if any, that satisfy the constraint \eqref{eq:extraconstraint}. This way we would be sure that we are dealing with honest solutions to $l_{P}$-corrected eleven-dimensional Supergravity and at the same time it would single out a preferred class of eliptically fibered Calabi-Yau manifolds. 

In this letter we are going to propose a simple class of non-geometric compactifications that directly evades the Maldacena-Nu\~nez theorem at the classical level, by a mechanism original of reference \cite{twistHeterotic}. Therefore, no $l_{P}$-corrections are needed to obtain non-trivial solutions, and thus no inconsistency arises, since there exist closed manifolds with non-trivial flux and warp factor that solve the equations of motion of the theory at the classical level. We don't want to imply with this that $l_{P}$-corrections are not relevant: they certainly are of utmost importance in order to understand String/M-theory backgrounds. However, we think that it may be a good idea to understand first non-geometric backgrounds without corrections, namely the \emph{zero-order} solution, before considering $l_{P}$-corrections to non-geometric backgrounds. The non-geometric solutions presented in this letter thus constitute the zero order non-geometric solution, which happens to be non-trivial, in the sense that it allows for non-trivial flux and warp-facor, in contrast to what happens in the geometric case. In any case, as we have said, ideally the ultimate goal would be to include and understand $l_{P}$-corrections for geometric as well as for non-geometric compactification backgrounds.

%%%%%%%%%%%%%%%%%%%%%%%%%%%%%%%%%%%%%%%%%%%%%%%%%%%%%%%%%%%%%%%%%%%%%%%%%%%%%%%%%%%%%%%%%%%%%%%%%%%%%%%%%%%%%%%%%%%%%%%%%%%%
%%%%%%%%%%%%%%%%%%%%%%%%%%%%%%%%%%%%%%%%%%%%%%%%%%%%%%%%%%%%%%%%%%%%%%%%%%%%%%%%%%%%%%%%%%%%%%%%%%%%%%%%%%%%%%%%%%%%%%%%%%%%
%%%%%%%%%%%%%%%%%%%%%%%%%%%%%%%%%%%%%%%%%%%%%%%%%%%%%%%%%%%%%%%%%%%%%%%%%%%%%%%%%%%%%%%%%%%%%%%%%%%%%%%%%%%%%%%%%%%%%%%%%%%%
%%%%%%%%%%%%%%%%%%%%%%%%%%%%%%%%%%%%%%%%%%%%%%%%%%%%%%%%%%%%%%%%%%%%%%%%%%%%%%%%%%%%%%%%%%%%%%%%%%%%%%%%%%%%%%%%%%%%%%%%%%%%

\section{A class of non-geometric M-theory compactification backgrounds}
\label{sec:Nogeometric}

%%%%%%%%%%%%%%%%%%%%%%%%%%%%%%%%%%%%%%%%%%%%%%%%%%%%%%%%%%%%%%%%%%%%%%%%%%%%%%%%%%%%%%%%%%%%%%%%%%%%%%%%%%%%%%%%%%%%%%%%%%%%
%%%%%%%%%%%%%%%%%%%%%%%%%%%%%%%%%%%%%%%%%%%%%%%%%%%%%%%%%%%%%%%%%%%%%%%%%%%%%%%%%%%%%%%%%%%%%%%%%%%%%%%%%%%%%%%%%%%%%%%%%%%%
%%%%%%%%%%%%%%%%%%%%%%%%%%%%%%%%%%%%%%%%%%%%%%%%%%%%%%%%%%%%%%%%%%%%%%%%%%%%%%%%%%%%%%%%%%%%%%%%%%%%%%%%%%%%%%%%%%%%%%%%%%%%
%%%%%%%%%%%%%%%%%%%%%%%%%%%%%%%%%%%%%%%%%%%%%%%%%%%%%%%%%%%%%%%%%%%%%%%%%%%%%%%%%%%%%%%%%%%%%%%%%%%%%%%%%%%%%%%%%%%%%%%%%%%%

In reference \cite{Shahbazi:2015mva} a \emph{twist} in the standard gluing of the local equations of motion of eleven-dimensional Supergravity on eight-manifolds was proposed, by means of the use of a particular atlas on the space-time manifold. In this section we are going to adopt a different point of view, proposing a slightly modified construction, which highlights the interpretation of such twisted supersymmetric compactification backgrounds as non-geometric compactification backgrounds. As a result, we will obtain that the total space of the non-geometric solution is still a manifold, although with different topological and geometric properties from the corresponding geometric solution, and that the Supergravity fields become tensors taking values on a particular line bundle.

\begin{remark}
The idea is to consider the local analysis of reference \cite{Becker:1996gj} and patch it globally in a non-trivial way by using not only local diffeomorphisms but also the trombone symmetry of the warp factor. We will see that when performing such a non-trivial patching the global space is still a manifold, but with very different geometric properties and topology from the standard solution of reference \cite{Becker:1996gj}.
\end{remark}

\noindent
The starting point is the standard one for compactification spaces. We will assume that the space-time manifold $\mathcal{M}$ can be written as a topologically trivial direct product

\begin{equation}
\label{eq:productmanifold}
\mathcal{M} = \mathbb{R}_{1,2}\times\mathcal{M}_{8}\, , 
\end{equation}

\noindent
where $\mathbb{R}_{1,2}$ is three-dimensional Minkowski space-time and $\mathcal{M}_{8}$ is an eight-dimensional, Riemannian, compact, oriented, spinnable manifold. According to the product structure (\ref{eq:productmanifold}) of the space-time manifold $\mathcal{M}$, the tangent bundle splits as follows\footnote{We omit the pull-backs of the canonical projections.}

\begin{equation}
\label{eq:tangentdecomposition}
T\mathcal{M} = \mathbb{R}_{1,2}\oplus T\mathcal{M}_{8}\, . 
\end{equation}

\noindent
Let $\mathcal{U} = \left\{U_{a}\right\}_{a\in I}$ be a good open covering of $M_{8}$. Then: 

\begin{equation}
\mathcal{V} = \mathbb{R}_{1,2}\times\mathcal{U} = \left\{V_{a} = \mathbb{R}_{1,2}\times U_{a}\right\}_{a\in I}\, ,
\end{equation}

\noindent
is a good open covering of $M$. We define in $M$ a family $\mathsf{g} = \left\{ \mathsf{g}_{a}\right\}_{a\in I}$ of local the Lorentzian metrics, where $\mathsf{g}_{a}$ is a locally defined metric on $\mathbb{R}_{1,2}\times U_{a}$, given by:

\begin{equation}
\label{eq:11dmetricproduct}
\mathsf{g}_{a} = \Delta^2_{a}\eta_{1,2}\times g_{8}|_{U_{a}}\, , 
\end{equation}

\noindent
where $g_{8}$ is a Riemannian metric on $\mathcal{M}_{8}$ and $\Delta_{a}\in C^{\infty}\left( U_{a}\right)$. Similarly, we define in $M$ a family $\mathsf{G} = \left\{ \mathsf{G}_{a}\right\}_{a\in I}$ of local closed four-forms, where $\mathsf{G}_{a}$ is a locally defined closed four-form on $\mathbb{R}_{1,2}\times U_{a}$, given by:

\begin{equation}
\label{eq:Gdecomposition}
\mathsf{G}_{a} = \mathrm{Vol}_{1,2}\wedge \xi_{a} +  G|_{U_{a}}\, , \qquad \xi_{a}\in\Omega^{1}\left( U_{a}\right)\, , \quad G\in\Omega^{4}_{cl}\left(\mathcal{M}_{8}\right) \, ,
\end{equation}

\noindent
where $\mathrm{Vol}$ is the standard volume form of Minkowski space. The idea now is to impose $(\mathsf{g}_{a},\mathsf{G}_{a})$ to be a local solution of the equations of motion of eleven-dimensional Supergravity for every $a\in I$. Then, we will patch this local solutions globally by using not only local diffeomorphisms but also a particular global symmetry of the equations of motion. As we will see in a moment, the global geometry of $M$ will depend on the specific patching used for the family of local solutions. More precisely, for each $a\in I$ of the good open cover $\mathcal{U} = \left\{U_{a}\right\}_{a\in I}$  of $M_{8}$, let us denote by:

\begin{equation}
Sol_{a} = (g_{8}|_{U_{a}}, G|_{U_{a}}, \Delta_{a}, \xi_{a})\, , \qquad \Delta_{a}\in C^{\infty}(U_{a}), \qquad \xi_{a}\in \Omega^{1}_{cl}(U_{a})\, ,
\end{equation}

\noindent
a local solution to the equations of motion of the theory, in the compactification background explained above. Notice that, in contrast to $\Delta_{a}$ and $\xi_{a}$, which are defined only locally, $g_{8}|_{U_{a}}$ and  $G|_{U_{a}}$ are just the restriction of the globally defined tensors $g_{8}$ and $G$ to $U_{a}$, so they are well-defined globally. Now, a standard compactification would construct a global solution to the equations of motion by globally patching the family of local solutions $\left\{Sol_{a}\right\}_{a\in I}$ using just local diffeomorphisms. This way we would obtain a globally well-defined metric $\mathsf{g}$ and four-form $\mathsf{G}$ on $M$. In contrast, a non-geometric compactification is characterized by patching-up local solutions by using not only local diffeomorphisms but also symmetries of the equations of motion. 

What was done in reference \cite{Shahbazi:2015mva} was to patch up the global solution using local diffeomorphisms and also a particular symmetry of the equations of motion: the \emph{trombone symmetry} of the warp factor, consisting on rescalings of the warp factor by a constant. In reference \cite{Shahbazi:2015mva} we used a very particular atlas in order to obtain that the Supergravity fields are tensors. We will drop here that condition and we will adopt the natural point of view of a non-geometric compactification: the global Supergravity fields obtained by the non-trivial patching of the local solutions may not be tensors but objects of a more general type. In our case we will obtain that the Supergravity fields are tensors valued on a particular line bundle $L$.  

Hence, the kind of compactification backgrounds described in \cite{Shahbazi:2015mva} are indeed non-geometric, although of a simple type, namely the symmetry used to patch-up the solution globally is a simple rescaling of the warp factor. Remarkably enough, the global space of the compactification is still a manifold, something that is not guaranteed for more general non-geometric compactifications. 

Let us do the global patching explicitly. Given the good open cover $\mathcal{U}$, for each $U_{a}\in\mathcal{U}$ we have a locally defined warp factor $\Delta_{a}\in C^{\infty}(U_{a})$. As we have said, two local warp factors $\Delta_{a}$ and $\Delta_{b}$, $a, b\in I$ are related by a rescaling of the warp factor on the non-empty intersection $U_{a}\cap U_{b} \neq \left\{\emptyset\right\}$ of $U_{a}$ and $U_{b}$. Then we have:

\begin{equation}
\label{eq:patchingDelta}
\Delta_{a} = \beta_{ab} \Delta_{b}\, , \qquad \beta_{ab}\in \mathbb{R}^{\ast}\, ,
\end{equation}

\noindent
which as we have said is a symmetry of the equations of motion, as it is required to obtain a global solution. Equation \eqref{eq:patchingDelta} implies that:

\begin{equation}
\beta_{ab} = \beta^{-1}_{ba}\, ,\qquad \beta_{ab}\beta_{bc}\beta_{ca} = 1\, ,
\end{equation}

\noindent
where the last equation holds on the triple non-empty triple intersection $U_{a}\cap U_{b}\cap U_{c} \neq  \left\{\emptyset\right\}$. Hence: 

\begin{equation}
L = \left(\left\{U_{a}\right\}, \beta_{ab}, \mathbb{R}\right)\, ,
\end{equation}

\noindent
defines a real line bundle $L$ over $M$. Hence, the warp factor cannot be described as a globally defined function on $M_{8}$. however, it does define a globally defined closed one form $\varphi\in\Omega^{1}(M_{8})$, given on every open set $U_{a}\in \mathcal{U}$ by:

\begin{equation}
\varphi|_{U_{a}} = d\log \, \Delta_{a}\, .
\end{equation}

\noindent
Therefore, we have obtained what is the global structure of the warp factor: after trivializing $L$, it is given as closed one-form by a connection on $L$. 

We have to patch-up now the local solutions $\left\{ Sol_{a}\right\}_{a\in I}$ of the the theory. We are not interested in patching up the most general local compactification background, but only the $\mathcal{N}=2$ supersymmetric compactification backgrounds of reference \cite{Becker:1996gj}. Therefore, each solution $Sol_{a},\, a\in I,\,$ will be a local solution of the type presented in \cite{Becker:1996gj}, namely globally conformal to a Calabi-Yau four-fold. Therefore, from reference \cite{Becker:1996gj} we obtain that:

\begin{equation}
Sol_{a} = (g_{8}|_{U_{a}}, G|_{U_{a}}, \Delta_{a}, \xi_{a})\, , \qquad \Delta_{a}\in C^{\infty}(U_{a}), \qquad \xi_{a}\in \Omega^{1}_{cl}(U_{a})\, ,
\end{equation}

\noindent
is equipped with a local $SU(4)$-structure $(J_{a},\omega_{a},\Omega_{a})$ satisfying:

\begin{equation}
\label{eq:localstructureeq}
\nabla_{a} J_{a} = 0\, , \qquad\nabla_{a}\omega_{a} = 0\, , \qquad \nabla_{a}\Omega_{a} = 0\, ,
\end{equation}

\noindent
where $J_{a}$ is a local complex structure, $\omega_{a}$ is a local symplectic structure, $\Omega_{a}$ is a local $(4,0)$-form and $\nabla_{a}$ is the locally-defined Levi-Civita connection associated to the local metric $g_{a} = \Delta_{a} g_{8}|_{U_{a}}$. In other words, $(J_{a},\omega_{a},\Omega_{a})$ is a local integrable $SU(4)$-structure. In addition, we have \cite{Becker:1996gj}:

\begin{equation}
\label{eq:DeltaMaxwelllocal}
\xi_{a} = d(\Delta^{3}_{a})\, , \qquad d\ast d\log \Delta_{a} + G|_{U_{a}}\wedge G|_{U_{a}} = 0\, , \qquad \omega_{a}\wedge G|_{U_{a}} = 0\, , \qquad G\in \Omega^{2,2}(U_{a})\, ,
\end{equation}

\noindent
where $G|_{U_{a}}$ is $(2,2)$ with respect to $J_{a}$.

\begin{remark}
As we explained in section \eqref{sec:Mtheory3d}, supersymmetric compactification backgrounds are time-like supersymmetric solutions, and thus it is enough to satisfy the Maxwell equations for $\mathsf{G}$ in order to satisfy all the equations of motion.
\end{remark}

\noindent
Using now that the global patching is performed by means of only local diffeomorphisms and the trombone symmetry, together with the results of reference \cite{Becker:1996gj}, we obtain that, for each $a\in I$, the local $SU(4)$-structure $(J_{a},\omega_{a},\Omega_{a})$ can be written as:

\begin{equation}
\label{eq:localtoglobal}
g_{8\, a} = \Delta_{a} g_{8}|_{U_{a}}\, , \qquad J_{a} =  J|_{U_{a}} \, , \qquad \omega_{a} = \Delta_{a} \omega|_{U_{a}}\, , \qquad \Omega_{a} = \Delta^{2}_{a} \Omega|_{U_{a}}\, .
\end{equation}

\noindent
where $(g_{8},J,\omega,\Omega)$ is a global $SU(4)$-structure on $M_{8}$, namely $J$ is an almost-complex structure, $\omega$ is the fundamental two-form and $\Omega$ is the $(4,0)$. In order to fully characterize the non-geometric compactification background, we have to obtain the geometry of $M_{8}$ from the local supersymmetry conditions \eqref{eq:localstructureeq}, \eqref{eq:DeltaMaxwelllocal} and \eqref{eq:localtoglobal}.

\begin{prop}
\label{prop:localtoglobal}
Equations \eqref{eq:localstructureeq}, \eqref{eq:DeltaMaxwelllocal} and \eqref{eq:localtoglobal} are equivalent to $(M_{8},J,g_{8})$ being an Hermitian manifold with integrable almost-complex structure $J$ which is equipped with a $SU(4)$-structure $(J, \omega, \Omega)$ such that:

\begin{equation}
\label{eq:lcKglobal}
\qquad d\omega = \varphi \wedge\omega\, , \qquad \nabla_{a}\Omega_{a} = 0\, .
\end{equation}

\noindent
and in addition

\begin{equation}
\label{eq:DeltaMaxwellglobal}
\xi_{a} = d(\Delta^{3}_{a})\, , \qquad d\ast d\varphi + G \wedge G = 0\, , \qquad \omega\wedge G = 0\, , \qquad G\in \Omega^{2,2}(M_{8})\, ,
\end{equation}

\noindent
Therefore, $M_{8}$ is a locally conformally K\"ahler manifold with Lee form $\varphi$ and and locally conformally parallel $(4,0)$-form $\Omega$.
\end{prop}

\begin{proof}
From \eqref{eq:localtoglobal} we see that the local complex structures $\left\{ J_{a}\right\}_{a\in I}$ patch up to a well-defined almost-complex structure $J$ in $M_{8}$. Writing the Nijenhuis tensor of $J$ as:

\begin{equation}
N|_{U_{a}}(u,v) = (\tilde{\nabla}^{a}_{u}J)(J v ) - (\tilde{\nabla}^{a}_{v} J)(J u) + (\tilde{\nabla}^{a}_{Ju}J)(v) - (\tilde{\nabla}^{a}_{Jv}J) (u)\, ,\qquad u,v\in\mathfrak{X}(M_{8})\, ,
\end{equation}

\noindent
we obtain that $N|_{U_{a}} = 0$ for every $U_{a}\in\mathcal{U}$, and thus $J$ is integrable and $(M_{8},g_{8},J)$ is a Hermitian manifold. In addition, $G$ is globally $(2,2)$ in $M_{8}$. Since $\omega_{a}$ is, for each $a\in I$, a rescaling of $\omega|_{U_{a}}$ we obtain that $\omega_{a}\wedge G|_{U_{a}} = 0$ and $G\in \Omega^{2,2}(U_{a})$ are equivalent to:

\begin{equation}
\omega\wedge G = 0\, , \qquad G\in \Omega^{2,2}(M_{8})\, .
\end{equation}

\noindent
Using now that $\varphi|_{U_{a}} = d\log\Delta_{a}$, we obtain that the global form of the equation of motion for the warp factor is:

\begin{equation}
d\ast \varphi + G\wedge G = 0\, . 
\end{equation}

\noindent
Since $J_{a}$ is a complex structure, we obtain that the condition $\nabla_{a}\omega_{a}$ is equivalent to $d\omega_{a} = 0$, which in turn is equivalent to:

\begin{equation}
d\omega = \varphi\wedge\omega\, .
\end{equation}

\noindent
For the converse, it is easy to see that equations \eqref{eq:lcKglobal} and \eqref{eq:DeltaMaxwellglobal} locally imply equations \eqref{eq:localstructureeq}, \eqref{eq:DeltaMaxwelllocal} and \eqref{eq:localtoglobal}.
\end{proof}

\noindent
Using now proposition \ref{prop:localtoglobal}, we have then proven the following theorem:

\begin{thm}
\label{thm:susynongeometric}
Let $M_{8}$ be an eight-dimensional compact manifold equipped with a $SU(4)$-structure $(J,\omega,\Omega)$ such that $J$ is integrable, $\omega$ is a locally conformally K\"ahler structure with Lee-form $\varphi$ and $\Omega$ is locally conformally parallel. Then, $(M_{8},J,\omega,\Omega)$ is a non-geometric admissible M-theory compactification background to three-dimensional Minkowski space-time provided that there exists a closed four-form $G\in \Omega^{4}(M_{8})$ such that:

\begin{equation}
\omega\wedge G = 0\, , \qquad G\in \Omega^{2,2}(M_{8})\, ,
\end{equation}

\noindent
and a solution to the equation of motion: 

\begin{equation}
\label{eq:globalwarpfactor}
d\ast \varphi + G\wedge G = 0\, . 
\end{equation}

\noindent
of the warp factor exists.
\end{thm} 

\noindent
The non-geometric background that we have obtained is very different from the standard Calabi-Yau compactification background, as a result of the non-trivial global patching. The topology of both manifolds is completely different, since a Calabi-Yau manifold is K\"ahler and a locally conformally K\"ahler manifold is not. In particular, the \emph{hodge diamonds} of both manifolds are completely different. Hence, we should expect the effective theories of the corresponding compactifications to be completely different too. In the next section we will indeed provide an explicit example that solves all the equations of motion for $\xi$ and $G$, giving thus a counterexample to the Maldacena-Nu\~nez no-go theorem. The Supergravity fields are no longer global tensors, but tensors taking values on the line bundle $L$. In fact, we have:

\begin{equation}
\mathsf{g} \in \Gamma(S^{2}T^{\ast}M_{8},L)\, , \qquad \xi\in \Omega^{1}(M_{8},L)\, .
\end{equation}

\noindent
To summarize, we have found a simple class of non-geometric M-theory backgrounds in which the total space is again a manifold and which:

\begin{itemize}

\item Need an underlying non-simply connected topological manifold. If the manifold is simply-connected the compactification background becomes geometric.

\item Evade the Maldacena-Nu\~nez no-go theorem.

\end{itemize}

\noindent
These are properties that are expected to be present in non-geometric backgrounds, as shown in reference \cite{twistHeterotic} in the Heterotic case. It is because of the second feature that we will be able to construct an explicit eight-dimensional non-geometric background which evades the Maldacena-Nu\~nez no-go theorem and thus evades any possible inconsistency coming from introducing $l_{P}$-corrections in the equations of motion but not in the classical Killing spinor equations.

\begin{remark}
A locally conformally K\"ahler $k$-fold equipped with a locally conformally parallel $(k,0)$-form was dubbed in \cite{Shahbazi:2015mva} a locally conformally Calabi-Yau manifold for obvious reasons. This is precisely what is required by supersymmetry in the class of compactifications considered in this note. We will refer then to the backgrounds specified in \ref{thm:susynongeometric} as locally conformally Calabi-Yau.
\end{remark}

%%%%%%%%%%%%%%%%%%%%%%%%%%%%%%%%%%%%%%%%%%%%%%%%%%%%%%%%%%%%%%%%%%%%%%%%%%%%%%%%%%%%%%%%%%%%%%%%%%%%%%%%%%%%%%%%%%%%%%%%%%%%
%%%%%%%%%%%%%%%%%%%%%%%%%%%%%%%%%%%%%%%%%%%%%%%%%%%%%%%%%%%%%%%%%%%%%%%%%%%%%%%%%%%%%%%%%%%%%%%%%%%%%%%%%%%%%%%%%%%%%%%%%%%%

\subsection{An explicit solution on a complex Hopf four-fold}
\label{sec:complexHopf}

%%%%%%%%%%%%%%%%%%%%%%%%%%%%%%%%%%%%%%%%%%%%%%%%%%%%%%%%%%%%%%%%%%%%%%%%%%%%%%%%%%%%%%%%%%%%%%%%%%%%%%%%%%%%%%%%%%%%%%%%%%%%
%%%%%%%%%%%%%%%%%%%%%%%%%%%%%%%%%%%%%%%%%%%%%%%%%%%%%%%%%%%%%%%%%%%%%%%%%%%%%%%%%%%%%%%%%%%%%%%%%%%%%%%%%%%%%%%%%%%%%%%%%%%%

The explicit solution that we will construct is a very particular case of a complex Hopf four-fold. Let us start though with a general construction.

A complex Hopf manifold $\mathbb{C}H^{m}_{\alpha}$ of complex dimension $m$ is the quotient of $\mathbb{C}^{m}\backslash \left\{ 0\right\}$ by the free action of the infinite cyclic group $\mathfrak{S}_{\alpha}$ generated by $z\to \alpha z$, where $\alpha\in\mathbb{C}^{\ast}$ and $0<|\alpha| <1$. In other words, it is $\mathbb{C}^{m}\backslash \left\{ 0\right\}$ quotiented by the free action of $\mathbb{Z}$, where $\mathbb{Z}$, with generator $\alpha$ acting by holomorphic contractions. The group $\mathfrak{S}_{\alpha}$ acts freely on $\mathbb{C}^{m}\backslash \left\{ 0\right\}$ as a properly discontinuous group of complex analytic transformations of $\mathbb{C}^{m}\backslash \left\{ 0\right\}$. Hence, the quotient space:

\begin{equation}
\mathbb{C}H^{m}_{\alpha} = \left( \mathbb{C}^{m}\backslash \left\{ 0\right\}\right)/\mathfrak{S}_{\alpha}\, ,
\end{equation}

\noindent
is a complex $m$-fold. It can be shown that complex $m$-dimensional Hopf manifolds $\mathbb{C}H^{m}_{\alpha}$ are diffeomorphic to $S^{1}\times S^{2m-1}$. As a result:

\begin{equation}
b^{1}\left( \mathbb{C}H^{m}_{\alpha}\right) = b^{2m-1}\left( \mathbb{C}H^{m}_{\alpha}\right) = 1\, ,
\end{equation}

\noindent
namely the first betti number is odd and hence $\mathbb{C}H^{m}_{\lambda}$ does not admit a K\"ahler metric. Notice the standard K\"ahler structure on $\mathbb{C}^{m}\backslash \left\{ 0\right\}$ does not descend to $\mathbb{C}H^{m}_{\alpha}$ since it is not $\mathfrak{S}_{\alpha}$-invariant. It admits however a locally conformally K\"ahler structure. To prove this, let us take $\mathbb{C}^{m}\backslash \left\{ 0\right\}$ equipped with the following metric and (1,1)-form: 

\begin{equation}
\label{eq:g0omega0}
g_{0} = \frac{dz^{t}\otimes d\bar{z}}{\bar{z}^{t}z}\, , \qquad\omega_{0} = i\frac{dz^{t}\wedge d\bar{z}}{\bar{z}^{t}z}\, . 
\end{equation}

\noindent
The (1,1)-form $\omega_{0}$ is not closed but it satisfies:

\begin{equation}
d\omega_{0} = \varphi_{0}\wedge\omega_{0}\, , 
\end{equation}

\noindent
where: 

\begin{equation}
\varphi_{0} = \frac{z^{t}d\bar{z}+\bar{z}^{t}dz}{\bar{z}^{t}z}\, ,
\end{equation}

\noindent
Since $g_{0}$, $\omega_{0}$ and $\varphi_{0}$ are invariant under $\mathfrak{S}_{\alpha}$ transformations, they descend to a well defined metric $g$ and (1,1)-form $\omega$ in $\mathbb{C}H^{m}_{\alpha}$, with corresponding Lee-form $\varphi$. In $\mathbb{C}^{m}\backslash \left\{ 0\right\}$ we have that $\varphi_{0}$ is exact, since $\varphi_{0} = d \log z^{t}\bar{z}$. This should be expected, as $(g_{0},\omega_{0})$ is globally conformal to the standard K\"ahler structure on $\mathbb{C}^{m}\backslash \left\{ 0\right\}$. However, $\varphi$ is not exact in $\mathbb{C}H^{m}_{\alpha}$, since $\log z^{t}\bar{z}$ is not globally well-defined in $\mathbb{C}H^{m}_{\alpha}$.

\subsubsection{The non-geometric solution:} Let us take now $m=4$, and $\alpha = \bar{\alpha}$. Then $\mathbb{C}H^{4}_{\alpha}$ is an eight-dimensional manifold of the type just described above. In particular, it is equipped with a locally conformally K\"ahler structure $(g,\omega)$ induced by the quotient of the $(g_{0},\omega_{0})$ given in equation \eqref{eq:g0omega0}. When $\alpha$ is real we can define in addition another globally defined (4,0)-form, induced by the following form on $\mathbb{C}^{m}\backslash \left\{ 0\right\}$:

\begin{equation}
\Omega_{0} = \frac{dz^{1}\wedge dz^{2}\wedge dz^{3}\wedge dz^{4}}{|z|^4}\, .
\end{equation}

\noindent
Now, since $\alpha$ is real, $\Omega_{0}$ is $\mathfrak{S}_{\alpha}$ invariant and therefore it descends to a well-defined (4,0)-form $\Omega$ on $\mathbb{C}H^{4}_{\alpha}$ satisfying:

\begin{equation}
\nabla_{a}\Omega_{a} = 0\, ,
\end{equation}

\noindent
and in particular:

\begin{equation}
d\Omega = 2\varphi\wedge \Omega\, .
\end{equation}

\noindent
Therefore $(g,\omega,\Omega)$ is precisely a locally conformally Calabi-Yau structure on $\mathbb{C}H^{4}_{\alpha}$, which is what was required by supersymmetry, see theorem \ref{thm:susynongeometric}. Hence, in order to obtain a full non-geometric solution, we just have to solve the equation of motion for $\mathsf{G} = \left\{ \mathsf{G}_{a}\right\}_{a\in I}$. Notice that, as we explained in section \ref{sec:Nogeometric}, local supersymmetry imposes:

\begin{equation}
\xi_{a} = d(\Delta^{3}_{a})\, ,
\end{equation}

\noindent
and that the only equation of motion that remains to be solved is the equation of motion for the warp factor, namely

\begin{equation}
d\ast\varphi + G\wedge G = 0\, .
\end{equation}

\noindent
In order to solve it, we are going to take $G = 0$. Notice that this does not trivialize the flux $\mathsf{G}$ since it is still non-zero and has \emph{one leg} on $M_{8}$. Taking $G = 0$ we obtain that the equation of motion for the warp factor reduces to:

\begin{equation}
\label{eq:coclosedvarphi}
d\ast\varphi = 0\, .
\end{equation}

\noindent
Since $\varphi$ is already closed $\varphi$ must be harmonic in order to solve equation \eqref{eq:coclosedvarphi}. It turns out that $\varphi$ is indeed harmonic; which, since it is already closed and it is the Lee-form of $\omega$, it is the same as requiring $g_{8}$ to be the Gauduchon metric. Therefore:

\begin{equation}
Sol = \left(\mathbb{C}H^{4}_{\alpha}, g_{8}, \omega, \Omega, \varphi\right)\, ,
\end{equation}

\noindent
is a compact \emph{non-geometric} solution of eleven-dimensional Supergravity with non-trivial flux and warp-factor. From a different point of view, one can see that $Sol$ is locally conformal to flat space equipped with the standard Calabi-Yau structure and therefore it trivially solves the supersymmetry equations. Globally however the geometry is very different and that in turn allows for the existence of a non-trivial flux and warp factor. We could say then that the non-trivial warp-factor and flux \emph{are supported} by the \emph{non-geometry} of the solution. 

\begin{remark}
In the standard compactification scenario, instead of $\varphi$ what we have is the derivative of the warp factor, say $df$, where $f\in C^{\infty}(M_{8})$ denotes the globally defined warp-factor on $M_{8}$. The equation of motion of the warp factor is given, after setting $G$ equals to zero, by:

\begin{equation}
\Delta\, f = 0\, .
\end{equation}

\noindent
Since $M_{8}$ is closed then $f$ must be constant and we obtain the famous Maldacena-Nu\~nez no-go theorem. In our non-geometric case however, we get a harmonic one-form $\varphi$, so as long as the first betti number of $M_{8}$ is bigger or equal than one, we are guaranteed to have at least one non-trivial solution. This is precisely what happens for the solution that we found, and that is why we are able to evade the no-go Maldacena-Nu\~nez theorem.
\end{remark}

\noindent
For completeness, let us write locally the warp factor and flux in local coordinates: let $(U_{a}, z_{a})$ be a local chart of $\mathbb{C}H^{4}_{\alpha}$. Then we have that: 

\begin{equation}
\varphi|_{U_{a}} = d \left(\log z^{t}_{a}\bar{z}_{a} + c^{\prime}_{a}\right)\, , \qquad c^{\prime}_{a}\in\mathbb{R}\, ,
\end{equation}

\noindent
and thus the warp factor of eleven-dimensional Supergravity compactified on $CH^{4}_{\lambda}$ is, at every coordinate chart $(U_{a}, z_{a})$, given by:

\begin{equation}
\Delta_{a} = c_{a} z^{t}_{a}\bar{z}_{a}\, , \qquad c_{a}\in \mathbb{R}^{\ast}\, .
\end{equation}

\noindent
Therefore, locally the four-form flux is given by:

\begin{equation}
\mathsf{G}|_{U_{a}} = \mathrm{Vol}_{1,2}\wedge d(c_{a} z^{t}_{a}\bar{z}_{a})^3\, .
\end{equation}

\noindent
\begin{remark}
In reference \cite{Shahbazi:2015mva}, a particular atlas was used in order to make $\left\{ \mathsf{G}_{a}\right\}_{a\in I}$ a globally defined tensor on $M$. However, from the point of view of a non-geometric compactification, we do not need to perform such a construction. For non-geometric compactifications the global objects that locally correspond to the fields of the theory are not expected to be standard tensors. In this case $\mathsf{G}$ can be understood as a four-form taking values on a real line bundle $L$:

\begin{equation}
\mathsf{G} = \mathrm{Vol}\wedge\xi\, , \qquad \xi\in \Omega^{1}(M_{8};L^3)\, .
\end{equation} 

\noindent
The real line bundle $L$ twists $\mathsf{G}$ from being a standard four-form and this is the result of the non-trivial global patching of the solution.
\end{remark}

\noindent
The solution $Sol = \left(\mathbb{C}H^{4}_{\alpha}, g_{8}, \omega, \Omega, \varphi\right)$ that we have obtained, although the simplest of its kind, has very interesting properties, some of them shared also by more general locally conformally K\"ahler manifolds. In particular, it is equipped with a holomorphic torus fibration and a transversely orientable, codimension one, real foliation with a $G_{2}$-structure on the leaves. Therefore, $Sol$ has the geometric properties found in \cite{Babalic:2014fua,Babalic:2014dea} for the most general $\mathcal{N}=1$ supersymmetric compactification of eleven-dimensional Supergravity to three-dimensional Minkowski space-time. This will be the subject of the next section.

%%%%%%%%%%%%%%%%%%%%%%%%%%%%%%%%%%%%%%%%%%%%%%%%%%%%%%%%%%%%%%%%%%%%%%%%%%%%%%%%%%%%%%%%%%%%%%%%%%%%%%%%%%%%%%%%%%%%%%%%%%%%
%%%%%%%%%%%%%%%%%%%%%%%%%%%%%%%%%%%%%%%%%%%%%%%%%%%%%%%%%%%%%%%%%%%%%%%%%%%%%%%%%%%%%%%%%%%%%%%%%%%%%%%%%%%%%%%%%%%%%%%%%%%%

\subsection{Foliations and principal torus fibrations on Vaisman manifolds}
\label{sec:Foltorus}

%%%%%%%%%%%%%%%%%%%%%%%%%%%%%%%%%%%%%%%%%%%%%%%%%%%%%%%%%%%%%%%%%%%%%%%%%%%%%%%%%%%%%%%%%%%%%%%%%%%%%%%%%%%%%%%%%%%%%%%%%%%%
%%%%%%%%%%%%%%%%%%%%%%%%%%%%%%%%%%%%%%%%%%%%%%%%%%%%%%%%%%%%%%%%%%%%%%%%%%%%%%%%%%%%%%%%%%%%%%%%%%%%%%%%%%%%%%%%%%%%%%%%%%%%

Let $(M,\omega)$ be a Vaisman manifold, namely $(M,\omega)$ is a locally conformally K\"ahler manifold with a parallel Lee-form $\theta$. Since $\theta$ is parallel, if it is non-zero at one point, it is non-zero at every point. Notice that the Hopf manifold that we found in section \ref{sec:complexHopf} that satisfy the local equations of motion of eleven-dimensional Supergravity is a particular example of Vaisman manifold. A Vaisman manifold $(M,\omega)$ is equipped with four canonical foliations, which are defined on $(M,\omega)$ by means of the Lee-form $\theta$ and the complex structure $J$ of $M$ as follows \cite{foliationslck,lcKbook}:

\begin{itemize}

\item $(M,\omega)$ is equipped with a completely integrable and regular codimension-one distribution $F\subset TM$, given by $\theta = 0$. We will denote by $\mathcal{F}$ the corresponding foliation, which is totally geodesic.

\item $(M,\omega)$ is equipped with a completely integrable and regular dimension one distribution $D\subset TM$ given by the vector field $v = \theta^{\sharp}$. We will denote by $\mathcal{D}$ the corresponding foliation, which is a geodesic foliation.

\item $(M,\omega)$ is equipped with a completely integrable and regular dimension one distribution $D^{\perp}\subset TM$ given by the vector field $w = J\cdot v$. We will denote by $\mathcal{D}^{\perp}$ the corresponding foliation. Notice that the distribution $D^{\perp}$ is perpendicular to $D$, and hence the symbol used.

\item $(M,\omega)$ is equipped with a completely integrable and regular dimension two distribution $T = D^{\perp}\oplus D\subset TM$. We will denote by $\mathcal{T}$ the corresponding foliation. The foliation $\mathcal{T}$ is a complex analytic foliation whose leaves are parallelizable complex analytic manifolds of complex dimension one. The leaves are totally geodesic, locally
Euclidean submanifolds of $M$ and the foliation is Riemannian.

\end{itemize}

\noindent
If the foliation $\mathcal{T}$ is regular, as it happens for the solution $Sol = \left(\mathbb{C}H^{4}_{\alpha}, g_{8}, \omega, \Omega, \varphi\right)$ that we found in section \ref{sec:complexHopf}, then the following result holds \cite{foliationslck,lcKbook}.

\begin{thm}
If the foliation $\mathcal{T}$ on a compact Vaisman manifold $(M,\omega)$ is regular then:

\begin{itemize}

\item The leaves are totally geodesic flat torii.

\item The leaf space $\mathcal{M} = M/\mathcal{F}$ is a compact K\"ahler manifold.

\item The projection $\pi$ is a locally trivial fibre bundle.

\end{itemize}
\end{thm}

\noindent
Therefore, compact Vaisman manifolds with regular foliation $\mathcal{T}$ are equipped with a non-trivial torus principal bundle over a K\"ahler manifold, in the line of the suggestion made in \cite{Grana:2014rpa}. This is interesting, because for F-theory applications one needs the eight-dimensional compactification manifold to admit a elliptic fibration, which must be singular to be non-trivial since the compactification space is an irreducible Calabi-Yau four-fold. In our case the fibration can be non-trivial yet non-singular, and that is indeed the case of the solution of section \ref{sec:complexHopf}. The interpretation, if any, of such non-singular and non-trivial fibrations in the context of F-theory remains unclear. This of course does not mean that there are no locally conformally K\"ahler four-folds admitting singular elliptic fibrations; this is currently an open problem\footnote{Private communication from Victor Vuletescu.}.

On the other hand, a compact Vaisman manifold admits a topological $Spin(7)$-structure, and in particular it is spin, as a consequence of having all the Chern numbers equal to zero. This $Spin(7)$-structure induces a $G_{2}$-structure on the leafs of the canonical foliation $\mathcal{F}$. If we restrict to the class of Hopf manifolds inside the class of Vaisman manifolds, then we have a very explicit result about the $G_{2}$-structure present in the leaves. Notice that the solution of section \ref{sec:complexHopf} is a Hopf manifold, so the following result applies \cite{foliationslck,lcKbook}.

\begin{prop}
\label{prop:Hopfspheres}
Let $(M,\omega)$ be a compact $2m$-dimensional Hopf manifold. Then, $\mathcal{F}$ is a totally geodesic foliation of $(2m-1)$-dimensional spheres defined through the diffeomorphism $M\simeq S^{1}\times S^{2m-1}$.
\end{prop}

\noindent
Let us apply proposition \ref{prop:Hopfspheres} to the $m=8$ case. Then the foliation $\mathcal{F}$ is by seven-dimensional spheres $S^{7}$. But a seven-dimensional sphere $S^{7}$ is equipped with a nearly parallel $G_{2}$-structure $\phi\in\Omega^{3}(S^{7})$, which satisfies:

\begin{equation}
d\phi = \tau_{0}\ast\phi\, , \qquad d\ast\phi = 0 \, , \qquad d\tau_{0} = 0\, .
\end{equation}

\noindent
Let denote by $\tau_{0}\in\Omega^{0}(S^{7}),\tau_{1}\in\Omega^{1}(S^{7}),\tau_{2}\in\Omega^{2}_{14}(S^{7})$ and $\tau_{2}\in\Omega^{3}_{27}(S^{7})$ the torsion classes of the $G_{2}$ structure $\phi$. Then, the $G_{2}$-structure $\phi$ satisfies $\tau_{2} = 0$ and it is therefore the codimension-one foliation of the non-geometric solution that we found is a particular case of the general characterization found in references \cite{Babalic:2014fua,Babalic:2014dea,Babalic:2015xia,Babalic:2015kka} for the most general eleven-dimensional Supergravity supersymmetric compactification background to three-dimensions. It is rewarding to see that although we are considering a non-geometric compactification background, the foliation structure of the most general geometric supersymmetric compactification background is preserved, which also indirectly shows that compactifying in this class of non-geometric compactification background should be possible in principle. 

%%%%%%%%%%%%%%%%%%%%%%%%%%%%%%%%%%%%%%%%%%%%%%%%%%%%%%%%%%%%%%%%%%%%%%
%%%%%%%%%%%%%%%%%%%%%%%%%%%%%%%%%%%%%%%%%%%%%%%%%%%%%%%%%%%%%%%%%%%%%%
%%%%%%%%%%%%%%%%%%%%%%%%%%%%%%%%%%%%%%%%%%%%%%%%%%%%%%%%%%%%%%%%%%%%%%
%%%%%%%%%%%%%%%%%%%%%%%%%%%%%%%%%%%%%%%%%%%%%%%%%%%%%%%%%%%%%%%%%%%%%% 

%\newpage
%\renewcommand{\leftmark}{\MakeUppercase{Bibliography}}
\phantomsection
\bibliographystyle{JHEP}
\bibliography{C:/Users/cshabazi/Dropbox/Referencias/References}
\label{biblio}
\newpage

%\bibliographystyle{habbrv} 
%\bibliography{C:/Users/cshabazi/Dropbox/Referencias/References}
\end{document}